\documentclass[A4,12pt]{article}
\usepackage{amsmath,amssymb,amsfonts,amsthm,graphicx}
\usepackage{color}
\usepackage{cite}
\usepackage{epsfig}
\usepackage{epstopdf}
\usepackage{footnote}
\usepackage{cite}
\usepackage{caption}
\usepackage{subcaption}
\usepackage{enumerate}

\usepackage[top=2cm, bottom=2cm, left=2cm, right=2cm]{geometry}

\newcommand{\spann}{\mathcal{SN}}

\newtheorem{theorem}{Theorem}[section]
\newtheorem{lemma}[theorem]{Lemma}

\theoremstyle{definition}
\newtheorem{definition}[theorem]{Definition}

\newtheorem{corollary}[theorem]{Corollary}
\newtheorem{example}[theorem]{Example}

\theoremstyle{remark}

\usepackage{algorithm}
\usepackage{algpseudocode}
\algnewcommand\algorithmicswitch{\textbf{switch}}
\algnewcommand\algorithmiccase{\textbf{case}}
\algnewcommand\algorithmicassert{\texttt{assert}}
\algnewcommand\Assert[1]{\State \algorithmicassert(#1)}%
\algdef{SE}[SWITCH]{Switch}{EndSwitch}[1]{\algorithmicswitch\ #1\ \algorithmicdo}{\algorithmicend\ \algorithmicswitch}%
\algdef{SE}[CASE]{Case}{EndCase}[1]{\algorithmiccase\ #1}{\algorithmicend\ \algorithmiccase}%
\algtext*{EndSwitch}%
\algtext*{EndCase}%

\title{Some Results on $[1, k]$-sets of Lexicographic Products of Graphs}
\author{P. Sharifani$^{1}$, M.R.~Hooshmandasl$^{2}$\\
	\footnotesize{$^{1,2}$Department of Computer Science, Yazd University, Yazd, Iran.}\\
	\footnotesize{$^{1,2}$The Laboratory of Quantum Information Processing, Yazd University, Yazd, Iran.}  \\
	\footnotesize{e-mail:$^1$ pouyeh.sharifani@gmail.com, $^2$ hooshmandasl@yazd.ac.ir}}
\date{}

\begin{document}
	
	\maketitle
	
	\begin{abstract}
		A subset $S \subseteq V$ in a graph $G = (V,E)$ is called  a $[1, k]$-set,	 if for every vertex $v \in V \setminus S$, $1 \leq \vert N_G(v) \cap S \vert \leq k$.		
		The $[1,k]$-domination number of $G$,  denoted by  $\gamma_{[1, k]}(G)$ is the size of the smallest $[1,k]$-sets of $G$.
		A set $S'\subseteq V(G)$ is called a total $[1,k]$-set, if for every vertex $v \in V$, $1 \leq \vert N_G(v) \cap S \vert \leq k$.
		If a graph $G$ has at least one total $[1, k]$-set then the cardinality of the smallest such set  is denoted by $\gamma_{t[1, k]}(G)$.
		We consider $[1, k]$-sets that are also independent. Note that not every graph has an independent $[1, k]$-set. For graphs having an independent $[1, k]$-set, we define  $[1, k]$-independence numbers which is denoted by $\gamma_{i[1, k]}(G)$.
		In this paper, we investigate the existence of $[1,k]$-sets in lexicographic products $G\circ H$. Furthermore, we  completely characterize graphs which their lexicographic product has  at least one total $[1,k]$-set. Also, we determine $\gamma_{[1, k]}(G\circ H)$, $\gamma_{t[1, k]}(G\circ H)$ and $\gamma_{i[1, k]}(G\circ H)$.
		Finally, we show that finding smallest total $[1, k]$-set is $NP$-complete.\\
		
	\end{abstract}

	\section{Introduction}
	The concept of domination and dominating set is a well studied topic in graph theory and has many extensions and applications \cite{haynes1998fundamentals, hedetniemi1991topics}.
	The problem of finding the smallest dominating set of a given graph $G$ is an $NP$-complete problem. Beside practical applications, this problem has many theoretical applications, e.g. in the theory of $NP$-completeness, many problems are reduced to this one.
	Its practical applications also include location problem, sets of representatives, monitoring communication, electrical networks, social network theory and so on \cite{haynes1998fundamentals, haynes1998domination, hedetniemi1991topics}.
	
	Many variants of dominations have been proposed  and  surveyed  in the literature such as total domination \cite{henning2013total}, efficient and open efficient 	dominations \cite{bange1988efficient,  gavlas2002efficient}, $k$-dominations \cite{chellali2012k, haas2014k}, rainbow domination \cite{bresar2008rainbow} and others like \cite{allan1978domination , yannakakis1980edge,Chang201289 , haynes1998paired, haynes1998fundamentals,Zhao201428}. Most of these problems are shown to be $NP$-hard  \cite{bange1988efficient, Chang201289, haynes1998paired, Zhao201428}.

	In graph theory, constructing complex graphs from some simpler ones is challenging, however it has many applications. So, studying properties of such complex  graphs and relations between properties of their components  is an interesting topic studing. This topic has led to many long-standing open problems such as Vizing's conjecture on the domination number of Cartesian products \cite{hartnell1991vizing}. Standard products such as Cartesian, lexicographic and strong products have been studied and applied widely in many areas such as group theory, expander graphs and graph-based coding theory schemes \cite{hammack2011handbook, imrichklav, hoory2006expander}. Moreover, various types for dominating sets of products of graphs were intensively investigated in \cite{nowakowski1996associative, henning2013total, dorbec2006some , kuziak2014efficient , gravier1997domination , rall2005total , hartnell2004dominating, li2009total}.
	
	Recently, Chellali et al. have studied $[j, k]$-sets \cite{chellali20131}, independent $[1, k]$-sets \cite{chellali2014independent} and proposed  total $[j,k]$-sets  in graphs.
	They have also pointed out a number of open problems on $[1,2]$-dominating sets in \cite{chellali20131}. Some of those problems are solved by X. Yang et al. \cite{yang20141} and AK. Goharshady et al. \cite{goharshady20161}.

	In this paper, we study of total $[1,k]$-sets and independent $[1,k]$-sets of lexicographic products of graphs.
	
	The rest of the paper is organized as follows:  In Section \ref{Terminology}, we review some necessary terminology and notation. In Section \ref{total}, we determine  $[1,k]$-domination, total $[1,k]$-domination number and $[1,k]$-independence numbers for some some classes of graphs such as paths and cycles. In Section \ref{lexproduct}, we study total $[1,k]$-sets  of lexicographic product of graphs and then, we completely characterize graphs which their lexicographic product has  at least one total $[1,k]$-set. Then, we determine the structure of all total $[1,k]$-sets for these graphs. Moreover, we generalize these results to independent $[1,k]$-sets. In Section \ref{lexproductno}, we  determine  $[1,k]$-domination number in lexicographic product of two given graph, and total $[1,k]$-domination number and  $[1,k]$-independence number in lexicographic product of graphs. In Section \ref{complextysec}, we prove that finding a total $[1,2]$-set with minimum cardinality for a graph is $NP$-complete.

	\section{Terminology and Notation}\label{Terminology}
	In this section, we minimally review some required terminology and notation of graph theory. For notation and terminology that are not defined here, we refer the reader to \cite{west2001introduction}. In this paper, $G$ is assumed to be a simple graph with vertex set $V(G)$ and edge set $E(G)$ of order $n=\vert V(G)\vert$. For a vertex $v \in V(G)$, the degree $d_{G}(v)$, or simply $d(v)$, of $v$ is the number of edges that are incident to $v$ in $G$. We denote the minimum and maximum degrees of vertices in $G$ by $\delta(G)$ and $\Delta(G)$, respectively.
	The open neighborhood $N_G(v)$ of a vertex $v\in V(G)$ equals $\{u :
	\{u,v\} \in E(G)\}$ and its closed neighborhood $N_G[v]$  is defined $N_G(v) \cup \{v\}$. For a simple and undirected graph like $G$, $d(v) = |N_G(v)|$. The open (closed) neighborhood of $S \subseteq V$ is defined to be the union of open (closed) neighborhoods of vertices in $S$ and is denoted by $N(S)$ ($N[S]$).

	Let $S \subseteq V$ and $v\in V$. The spanning number of $v$ with respect to $S$, denoted as $\spann_{S}(v)$, is defined $\vert N_G(v)\cap S\vert $. Whenever there is no risk of misunderstanding, we omit $S$ and simply use  $\spann(v)$
	\\
	A set $S \subseteq V$,  is called a $k$-dependent set of $G$, if for each vertex $v \in S$, t is the case that $\vert N_G(v) \cap S\vert\leq k$. For $k=0$, $S$ is called a $0$-dependent or independent set of $G$.

	A set $D \subseteq V$ is called a dominating set of $G$ if for every $v \in V \setminus D$, there exists some vertex $u\in D$ such that $v \in N(u)$.
	The domination number of $G$ is the minimum number among cardinalities of all dominating sets of $G$ and is denoted by $\gamma(G)$.
	A set $D \subseteq V$ is called a total dominating set of $G$ if for every $v \in V$, there exists some vertex $u\in D$ such that $v \in N(u)$. Total domination number is the minimum number among cardinalities of all  total dominating sets of $G$ and is denoted by $\gamma_{t}(G)$.

	For two given integers $j$ and $k$ such that $j \leq k$, a subset $D \subseteq V$ is called a $[j, k]$-set (total $[j, k]$-set) if for every vertex $v \in V\setminus D$ ($v \in V$), $j \leq |N(v) \cap D| \leq k$. In other words, each vertex $v \in V \setminus D$ ($v \in V$) is adjacent to at least $j$ but not more than $k$ vertices in $D$.
	The $[j, k]$-domination number  (total $[j, k]$-domination number) of $G$ is the minimum number among cardinalities of $[j, k]$-sets (total $[j, k]$-sets) of $G$ and is denoted by $\gamma_{[j,k]}(G)$ ($\gamma_{t[j,k]}(G)$). Note that total $[j,k]$-sets might not exist for an arbitrary graph.
	\\
	The family of all graphs like $G$ which have at least one $[j,k]$-set is denoted by  $\mathcal{D}_{[j,k]}$. Similarly  the class of all graphs $G$ with at least one total $[j,k]$-set is denoted by $\mathcal{D}^t_{[j,k]}$.
		Other types of dominating sets, that we are used in this work are summarized  in the Table \ref{table1}.

	\begin{table}[h]
		\begin{center}
			\caption{Some types of domination studied in this paper where $S\subset V$ .}\label{table1}
			\begin{tabular}{lll}
				\hline
				Name &  $v \in V\setminus S$ &  $v \in S$\\
				\hline
				
				$[1,k]$-set & $1\leq \vert N(v)\cap S \vert \leq k$ & -\\
				Independent $[1,k]$-set & $1\leq \vert N(v)\cap S \vert \leq k $ & $ \vert N(v)\cap S \vert =0$  \\
				$j$-dependent $[1,k]$-set & $1\leq \vert N(v)\cap S\vert \leq k$ & $0\leq \vert N(v)\cap S \vert \leq j$ \\
				Total $[1,k]$-set &  $1\leq\vert N(v)\cap S \vert \leq k$ & $1\leq \vert(v)\cap S \vert  \leq k$\\
				$j$-dependent total $[1,k]$-set &  $1\leq \vert N(v)\cap S\vert  \leq k$ & $1\leq \vert N(v)\cap S \vert \leq j$\\
				Efficient dominating &  $\vert N(v)\cap S \vert =1$ & $ \vert N(v)\cap S\vert=0$\\
				Open efficient dominating & $\vert N(v)\cap S\vert =1$ & $\vert N(v)\cap S\vert=1$ \\
			\end{tabular}
			
		\end{center}
	\end{table}

	\section{$[1,k]$-sets of Paths and Cycles } \label{total}
	In this section, first  we express some results about $[1,k]$-sets for paths and  cycles. Then, we determine $\gamma_{[1, k]}$, $\gamma_{t[1, k]}$,  $\gamma_{i[1, k]}$ for these graphs.

	\begin{lemma}\label{pathlem0}
		For every connected graph $G=(V,E)$  of order $n$ such that $\Delta(G) \geq n-2$, $\gamma_{t[1, 2]}(G)=\gamma_{t[1, k]}(G)=2$.
	\end{lemma}
	\begin{proof}
		Let $v\in V$ of degree $n-2$ and there is a vertex $u\in V$, such that $\{v,u\}\notin E$. Then, there exist a vertex $w$ such that $\{v,w\}\notin E$ and $\{w,u\}\notin E$. So $S=\{v,w\}$ is a total $[1,k]$-set for $G$.
	\end{proof}

	\begin{lemma}\label{pathlem00}
		Let $G$ be a connected graph with two adjacent vertices $u$ and $v$ such that $N[v]\cup N[u]=V(G)$. Then for any $k\geq 2$,  $D=\{u,v\}$ is a total $[1, k]$-set for $G$ and  $\gamma_{t[1, 2]}(G)= \gamma_{t[1, k]}(G)= 2$.
	\end{lemma}	
	\begin{proof}
		By lemma \ref{pathlem0}, the proof is clear.
	\end{proof}
	\begin{lemma}\label{pathlem000}
		Let $T$ be a tree such that the complement of $T$, $\overline{T}$ is a connected graph.  Then, $\gamma_{t[1, 2]}(\overline{T})=\gamma_{t[1, k]}(\overline{T})=2$.
	\end{lemma}
	\begin{proof}
		Since each tree like $T$ has at least two leaves like $v$ and $u$. Then  $\{v,u\}\in E(\overline{T})$ and  $N[v]\cup N[u]=V(\overline{T})$. So  $\gamma_{t[1, 2]}(\overline{T})=\gamma_{t[1, k]}(\overline{T})=2$.
	\end{proof}
	\begin{lemma}\label{pathlem}
		
		Let $G$ be a non-trivial path $P_n$ or a cycle $C_n$. Then
		
		\begin{equation*}
		\gamma_{t[1, k]}(G) = \left\{
		\begin{array}{ll}
		\frac{n}{2}\; & \text{if } n\equiv 0\mod 4\\
		\frac{n+1}{2}\; & \text{if } n \equiv 1 \mod 4\\
		\frac{n+2}{2}\; & \text{if } n\equiv 2 \mod 4\\
		\frac{n+1}{2}\; & \text{if } n \equiv 3 \mod 4,\\
		\end{array} \right.
		\end{equation*}
		and
		\begin{equation*}
		\gamma_{[1, k]}(G) =\gamma_{i[1, k]}(G)= \lceil \frac{n}{3}\rceil.
		\end{equation*}
	\end{lemma}
	
	\begin{proof}
		Let  $n > 2$. Since $\Delta(P_n)=\Delta(C_n)=2$, then total dominating sets are total $[1,2]$-sets, too.
		Obviously, we have $\gamma_{t[1,k]}(G)=\gamma_{t[1,2]}(G)=\gamma_{t}(G)$. It is easy to see that   $ \gamma_{t}(G)$ is equal to the claimed amount.
		Moreover, each $\gamma$-set for a path or cycle is independent, so
		$$\gamma_{i[1,k]}(G)=\gamma_{[1,k]}(G)=\gamma_{[1,2]}(G)=\gamma(G)= \lceil \frac{n}{3}\rceil.$$
	\end{proof}

\section{Total $[1,2]$-sets of Lexicographic Products of Graphs} \label{lexproduct}
The lexicographic product  of graphs $G$ and $H$, denoted by $G\circ H$ is a graph with the vertex set  $V(G \circ H) = V(G) \times V(H)$ and two vertices $(g,h)$ and $(g',h')$ are adjacent in $G \circ H$ if and only if either $\{g,g'\}\in E(G)$ or $g=g'$ and $\{h,h'\} \in E(H)$.
\\
 Note that if $G$ is not connected, then $G \circ H$ is not connected, too. So in this section, we always assume that $G$ is a connected graph.
\\
 In this section, we investigate properties of graphs $G$ and $H$ such that $G \circ H$ has a total $[1,2]$-set. Then we extend these results to total $[1,k]$-set. 
 Note that, it is possible that  $G\in \mathcal{D}^t_{[1,2]}$, however  $G \circ H\notin \mathcal{D}^t_{[1,2]}$, or vice versa.
\begin{example}
 Let $G_1, G_2$ and $G_3$ be graphs that are shown  in Figure \ref{fig:lex}.
 It can be verified that $S=\{a,f,g,h\}$ is an efficient dominating set of $G_1$ and $G_1 \notin \mathcal{D}^t_{[1,2]}$.  Let $S'$ be a $\gamma_{t[1,2]}$-set of size two for an arbitrary  graph like $H \in \mathcal{D}^t_{[1,2]}$. Then, $S\times S'$ is a $\gamma_{t[1,2]}$-set for  $G_1 \circ H$.
 \\
It is easy to see that $G_2 \notin \mathcal{D}^t_{[1,2]}$. As we will show in Theorem \ref{lexthem},  there is not any graph $H$ such that $G_2 \circ H \notin \mathcal{D}^t_{[1,2]}$.
  \\
  $S=\{u_1,u_3,u_5\}$ is a total ${[1,2]}$-set for $G_3$. If $H$ contains an isolated vertex such as $v_i$, then $S\times \{v_i\}$  is a total $[1,2]$-set for $G_3 \circ H$.

 \begin{figure}[h!]
   \centering
     \includegraphics[width=0.65\textwidth]{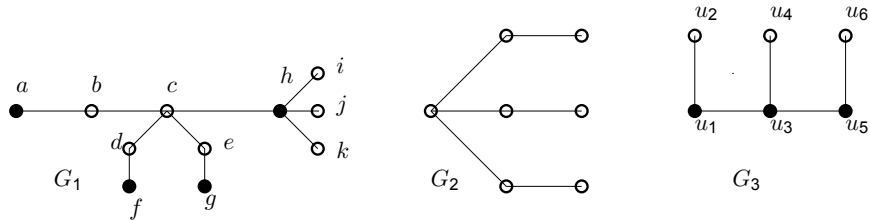}
      \caption{Graphs $G_1,G_3$ and $G_3$ where$G_1 , G_2 \notin \mathcal{D}^t_{[1,2]}$ and $G_3 \in \mathcal{D}^t_{[1,2]}$}
       \label{fig:lex}
\end{figure}
 \end{example}

\begin{definition}
Let $H$ and $G$  be graphs. The sets $G^{h_0}=\{(g,h_0)\in V(G \circ H): g\in V(G)\}$ and  $H^{g_0} = \left\{ (g_0,h) \in V(G \circ H) : h \in V(H) \right\}$ are called $G_{-}$Layer and $H_{-}$Layer respectively.
\end{definition}
\begin{lemma}\label{Lexneighbour}
Let $v$ and $v'$ be two adjacent vertices of $G$ and $u,u' \in V(H)$. Then
\begin{equation*}\label{BBB}
\begin{array}{ll}
N_{G \circ H}((v,u))\cup N_{G \circ H}((v',u))&=N_{G \circ H}((v,u'))\cup N_{G \circ H}((v',u'))\\&
=N_{G \circ H}((v,u))\cup N_{G \circ H}((v',u')).
\end{array}
\end{equation*}
\end{lemma}

\begin{proof}

We know that
\begin{equation*}
N_{G \circ H}((v,u))= \bigcup_{v_i \in N_G(v)}V(H^{v_i}) \cup \{(v,u_j): u_j \in N_H(u)\},
\end{equation*}
so
\begin{equation}\label{BB1}
\begin{array}{ll}
N_{G \circ H}((v,u)) \cup N_{G \circ H}((v',u'))= & (\bigcup_{v_i \in N_G(v)}V(H^{v_i})) \cup \{(v,u_j) : u_j \in N_H(u)\} \cup \\ &
(\bigcup_{v_i \in N_G(v')}V(H^{v_i})) \cup \{(v',u_j): u_j \in N_H(u')\}.
\end{array}
\end{equation}
It is easy to see that
\begin{equation}\label{BB2}
\{(v, u_j) : u_j \in N_H(u)\}\subseteq V(H^v).
\end{equation}
and
\begin{equation}\label{BB3}
\{(v',u_j): u_j \in N_H(u')\}\subseteq V(H^{v'}).
\end{equation}
By hypotheses  $\{v,v'\} \in E(G)$, we have

\begin{equation}\label{BB4}
\begin{array}{l}
V(H^v) \subseteq  N_{G \circ H}((v',u')),\\
V(H^{v'}) \subseteq  N_{G \circ H}((v,u)).
\end{array}
\end{equation}
 So by Relations \ref{BB1}, \ref{BB2}, \ref{BB3} and \ref{BB4}, it is implied that
$$N_{G \circ H}((v,u)) \cup N_{G \circ H}((v',u'))=
\bigcup_{v_i \in N_G(\{v, v'\})}V(H^{v_i}).$$
The  equation above shows that the union of neighbors of the vertices $(v,u)$ and $(v',u')$ is independent from $u$ and $u'$. Therefore, we have
$$N_{G \circ H}((v,u))\cup N_{G \circ H}((v',u))=N_{G \circ H}((v,u'))\cup N_{G \circ H}((v',u'))=N_{G \circ H}((v,u))\cup N_{G \circ H}((v',u')).$$
\end{proof}


\begin{lemma}\label{Lex3}
Let $D$ be a total $[1,2]$-set for $G \circ H \in \mathcal{D}^t_{[1,2]}$ which  contains more than two vertices of a $H_{-}$Layer $ H^v$. Then $G=K_1$ and $H \in \mathcal{D}^t_{[1,2]}$.
\end{lemma}

\begin{proof}
	Suppose $D$ be a total $[1,2]$-set of $G \circ H$ that contains vertices $(x,v), (y,v)$ and $(z,v)$ where $v\in V(G)$ and $x,y,z \in V(H)$. If there exists a vertex $v'\in V(G)$ such that $\{v,v'\}\in E(G)$, then all vertices of $H^{v'}$ are dominated by three vertices $(x,v), (y,v)$ and  $(z,v) $. This is a contradiction. So there is not any vertex adjacent to $v$.
	Since $G$ is a connected graph, $G=K_1=(\{v\},\emptyset)$ and $S=\{u\;:\;(v,u)\in D\}$ is a total $[1,2]$-set for $H$ and hence $H \in  \mathcal{D}^t_{[1,2]}$.
\end{proof}

Let $G$ be a nontrivial connected  graph and $G \circ H \in \mathcal{D}^t_{[1,2]}$. Then, every total $[1,2]$-set of  $G \circ H$ has  at most two vertices of each $H_{-}$Layer. For a total $[1,2]$-set $D$, we define
 $A_1^D$ as $\{(v,u):\vert V(H^v)\cap D\vert =1\}$ and $A_2^D$ as $\{(v,u):\vert V(H^v)\cap D\vert =2\}$.
The set $D$ satisfies in one of the following conditions:
\begin{enumerate}
  \item [1)]
  $A_1^D=\emptyset$,

\item [2)]
 $A_1^D \neq \emptyset$ and $A_2^D\neq \emptyset$,

 \item [3)]
 $A_2^D=\emptyset$.

\end{enumerate}

\begin{lemma}\label{Lex1}

Let $D$ be a total $[1,2]$-set of $G \circ H\in \mathcal{D}^t_{[1,2]}$ such that $A_2^D= \emptyset$. Then, $S=\{u:(u,v)\in D\}$ is a total $[1,2]$-set for $G$. In addition, if there is a  vertex $u \in S$ such that $\vert N(u) \cap S\vert = 2$;
 then $H$ contains an isolated vertex.
\end{lemma}

\begin{proof}
The proof is by contradiction. Assume $D$ is a total $[1,2]$-set of $G\circ H$ with $A_2^D=\emptyset$ and $S=\{u:(u,v)\in D\}$ is not a total set of $G$. Then, we have  three cases to consider.

	\begin{enumerate}

\item
There exists a vertex like $u\in S$ such that $\vert N(u) \cap S \vert=0$. It means that there is no vertex $u' \in N_G(u)$ such that $u'\in S$. The set $D$ is a total $[1,2]$-set and $u\in S$, so there exists a vertex $v\in V(H)$ such that $(u,v)\in D$. Similarly there exists a vertex $v' \in V(H)$ such that $(u,v')\in D$. This is a contradiction against $A_2^D=\emptyset$.
\item
There exists a vertex like $w\in V(G)\setminus S$ such that $\vert N_G(w) \cap S \vert=0$. Then, there is no vertex like $v\in V(H)$ such that $(u,v)\in D$. Moreover, there is no vertex $w' \in N_G(w)$ such that $w'\in S$.  Therefore vertices of $H^w$ can not be dominated by any vertex in $D$, which is a contradiction.
 \item
There exists a vertex like $w\in V(G)\setminus S$ such that $\vert N(w) \cap S \vert >2$. Then, there are at least three distinct vertices $w',w'',w'''\in N_G(w)\cap S$. By the definition of   $S$, there are vertices $v',v'',v''' \in V(H)$ such that $(w',v'),(w'',v''), (w''',v''')\in D$. These vertices dominate all vertices of $H^w$, which is a contradiction.
\end{enumerate}
\end{proof}

\begin{lemma}\label{Lex1and2a}
Let $G \circ H \in \mathcal{D}^t_{[1,2]}$ and $H$ does not contain any isolated vertex. Then, there exists either a $1$-dependent total $[1,2]$-set for $G$ or for each total $[1,2]$-set $D$ of $G$, $A_1^D=\{(v,u):\vert V(H^v)\cap D\vert=1\}\neq \emptyset$ and $A_2^D=\{(v,u):\vert V(H^v)\cap D\vert =2\} \neq \emptyset $.
\end{lemma}
\begin{proof}
Let  $D$ be a total $[1,2]$-set  of $G \circ H $ which contains at most one vertex from each $H_{-}$Layer. Since $H$ does not contain any isolated vertex then by Lemma \ref{Lex1} there is a $1$-dependent  total $[1,2]$-set like $S$ for $G$ such that $S=\{v: (v,u)\in D\}$  and  $A_2^D=\emptyset$.
\end{proof}

 For a given graph $G \circ H \in \mathcal{D}^t_{[1,2]}$ and a total $[1,2]$-set $D$ of  $G \circ H$ where $A_2^D\neq \emptyset$, we define the set $B^D$ as
 $B^D=\{\{u',u''\}:(v,u'),(v,u'')\in A_2^D\}.$
\begin{lemma}\label{Lex1and2}
Let $G \circ H \in \mathcal{D}^t_{[1,2]}$ where $H$ does not contain any isolated vertex and for any total $[1,2]$-set $D$ of $G\circ H$, $A_1^D\neq \emptyset$ and $A_2^D\neq \emptyset $. Then,  the following conditions  hold:
\begin{enumerate}
  \item [1)]
 Every element of  $B^D$ is a total $[1,2]$-set for $H$.
\item [2)]
The set $S'=\{v: (v,u) \in D\}$ is a $1$-dependent $[1,2]$-set for $G$.
 \item [3)]
 If there is a vertex $v\in S'$ such that $\vert N(v)\cap S'\vert =0$ then $dist_G(v,v')\geq 3$ for every $v' \in S'\setminus \{v\}$.
 \end{enumerate}
 \end{lemma}

\begin{proof}

Let $D$ be a total $[1,2]$-set of $G \circ H \in \mathcal{D}^t_{[1,2]}$; there are three cases to consider.
 \begin{itemize}
  \item [1)]
Suppose that $S=\{u^{\star},u^{\bullet}\} \in B$ is not a total $[1,2]$-set for $H$. Then two cases occur and in each case, we can establish  a contradiction with  $D$ is a total $[1,2]$-set.
  \begin{itemize}
  \item
  Let $\{u^{\star},u^{\bullet}\}\notin E(H)$ and there is a $(v',u')\in D$ such that $\{(v,u^{\star}),(v',u')\}\in E(G\circ H)$. Since $H$ dose not contain any isolated vertex, so any vertex $u''\in N_H(u')$ is dominated by $(v',u'),(v,u^{\star})$ and $(v, u^{\bullet})$.
  \item
  Let $\{u^{\star},u^{\bullet}\}$ does not dominate all vertices of $V(H)$. So, there is a vertex $(v',u')\in D$ such that $\{v,v'\} \in E(G)$ and  $(v',u')$  dominates all vertices of $H^v$. Then any vertex $u''\in N_H(u')$ is dominated by $(v',u'),(v,u^{\star})$ and $(v, u^{\bullet})$.
  \end{itemize}
\item [2)]
Suppose that  $S'=\{v: (v,u) \in D\}$  is not a $1$-dependent $[1,2]$-set for $G$. Then, three cases occur and in each case, we have a contradiction with  $D$ being a total $[1,2]$-set.
  \begin{itemize}
  \item
 There is a vertex  $v\in S'$ that is dominated by at least two vertices $v',v''\in S'$. So there are vertices $u,u',u''\in V(H)$ such that $(v,u),(v',u'),(v'',u'') \in D$. Since $H$ dose not contain any isolated vertex, there is a vertex $u'''\in V(H)$ such that $\{u,u'''\}\in E(H)$. Then, $(v,u''')$ is dominated by $(v,u),(v',u'),(v'',u'')$.
  \item
  There is a vertex  $v\in V(G)\setminus S'$ such that  $\vert N_G(v)\cap S'\vert=0$. So no vertex of $H^v$  is dominated by $D$.
  \item
  There is a vertex  $v\in V(G)\setminus S'$ such that $\vert N_G(x)\cap S'\vert > 2$. Then there are at least three vertices distinct $v',v'',v'''\in S'$ to dominate $v$. By definition of  $S'$,  there are vertices $u',u'',u'''\in V(H)$ such that $(v',u'),(v'',u''),(v''',u''')\in D$. These vertices dominate all vertices of $H^v$.
  \end{itemize}
 \item [3)]
 Let $v\in S'$ such that $\vert N(v)\cap S'\vert =0$ and there is a vertex  $v' \in S'$ such that  $dist_G(v,v')=2$.

By $\vert  N(v)\cap S'\vert  =0$, there exist  vertices $u',u'' \in V(H)$ such that $(v,u'), (v,u'') \in D$ and $ \{u',u''\} \in E(H)$. Suppose there is a vertex  $v' \in S'$ such that $dist_G(v,v')= 2$. So, there is  a vertex $v''\in V(G)$ such that $\{v,v''\},\{v',v''\} \in E(G)$. The vertices $(v,u')$, $(v,u'')$ and $(v',u')$ dominate all vertices of $H^{v''}$. It is contradictory with  $D$ being a total $[1,2]$-set. So we have $dist_G(v,v')\geq 3$.
 \end{itemize}

\end{proof}
\begin{lemma}\label{Lex2a}
Let $D$ be a total $[1,2]$-set of $G \circ H\in \mathcal{D}^t_{[1,2]}$ such that $A_1^D=\emptyset$. Then $S'=\{v: (v,u) \in D\}$ is an efficient dominating set of $G$.
\end{lemma}

\begin{proof}
Since $D$ be a total $[1,2]$-set of $G \circ H$, then there is a vertex $v\in S'$ such that the set $D$ contains $(v,u'),(v,u'')$ for some vertex $u',u''\in V(H)$. By  Lemma \ref{Lex1and2}, $\{u',u''\}$ is a total $[1,2]$-set for $H$. So for any vertex $v'\in N_G(v)$, none of vertices in $H^{v'}$ cannot be contained in $D$. Thus  $dist_G(v,v')\geq 3$ and $S$ is an efficient dominating set of $G$.

\end{proof}
In the sequel  $\mathcal{SD}_{[i,j]}^{k}(G)$ is  used to denote the  set of all $k$-dependent $[i,j]$-set $S$ of $G$ such that $S$ satisfies in the following condition
$$(\forall v\in S \,\,\,\, |N(v)\cap S|=0)\rightarrow (\forall v'\in S\setminus\{v\}\,\,\,\,  d(v,v')\geq3).$$
\begin{corollary}\label{Lex2}
Let $G $ be a connected nontrivial graph and $D$ be a total $[1,2]$-set of $G \circ H\in \mathcal{D}^t_{[1,2]}$, one of the following cases holds:
\begin{itemize}
\item
If $A_1^D=\{(u,v):\vert V(H^v)\cap D\vert=1\}=\emptyset$, then there is a total $[1,2]$-set $S=\{ u^{\star},u^{\bullet}\}$  in $H$ and an efficient dominating set $S'$ in $G$ such that $D'=S' \times S$ is a total $[1,2]$-set for $G \circ H$ and $\vert D\vert = \vert D' \vert = 2\vert S'\vert$.
\item
If $A_2^D=\{(u,v):\vert V(H^v)\cap D\vert=2\}=\emptyset$ and $H$ contains an isolated vertex $v$. Then there is a total $[1,2]$-set $S$ in $G$ where $D'=S \times \{v\}$ and $D'$  is a total $[1,2]$-set for $G \circ H$.  Moreover, we have $\vert D\vert = \vert D' \vert = \vert S\vert$.
\item
If $A_2^D=\{(u,v):\vert V(H^v)\cap D\vert=2\}=\emptyset$ and $H$ does not contain any isolated vertex, then for every vertex $v\in V(H)$  there is a $1$-dependent total $[1,2]$-set $S$ in $G$ such that $D'=S \times \{v\}$ and $D'$  is a total $[1,2]$-set for $G \circ H$.  Clearly, $\vert D\vert = \vert D' \vert = \vert S\vert$.
\item
If $A_1^D \neq \emptyset$ and $A_2^D \neq \emptyset$, then there is a total $[1,2]$-set $S=\{ u^{\star},u^{\bullet}\}$  in $H$ and  a $1$-dependent total $[1,2]$-set $S'$ in $G$ such that  for any vertex $v\in S$ and $u\in X$ where $X=\{x:\vert N_G(x) \cap S' \vert =0\}$, $dist(v,u)\geq 3$. Moreover $D'=((X \times  S)\cup (S'\setminus X)\times \{u^{\star}\})$ is a total $[1,2]$-set of size $\vert D\vert$  in $G \circ H$ and $\vert D\vert = \vert D' \vert = \vert S'\vert+ \vert X\vert$.
\end{itemize}
\end{corollary}
\begin{proof}
This corollary is a direct result of Lemma \ref{Lexneighbour}, \ref{Lex1}, \ref{Lex1and2} and \ref{Lex2a}.

\end{proof}

\begin{theorem} \label{lexthem}
	Let $G$ and $H$ be two graphs. Then, $G \circ H \in \mathcal{D}^t_{[1,2]}$   if and only if one of the following conditions holds:
\begin{enumerate}
\item
$G=K_1$ and $H \in \mathcal{D}^t_{[1,2]}$;
  \item
  $G$ has a total $[1,2]$-set $S$ such that if $S$ has a vertex $v$ where $\vert N(v) \cap S\vert =2$ then $H$ has an isolated vertex;
  \item
  $G$ is an efficient domination graph and $\gamma_{t[1,2]}(H)=2$;
  \item
  $\mathcal{SD}^1_{[1,2]}(G)\neq \emptyset$ and $\gamma_{t[1,2]}(H)=2$.
\end{enumerate}

\end{theorem}

\begin{proof}
	Suppose that $D$ be a total $[1,2]$-set of  $G \circ H \in \mathcal{D}^t_{[1,2]}$. If $D$ contains more than two vertices of a $H_{-}$Layer, then by Lemma \ref{Lex3}, $G=K_1$ and $H \in \mathcal{D}^t_{[1,2]}$.
	\\
	 If $D$ contains at most two vertices of each $H_{-}$Layer, then  there is a total $[1,2]$-set $D'$ for $G \circ H $ such that $\vert D' \vert=\vert D \vert$ and  vertices of $D'$ have  been choosen from two $G_{-}$Layers as $G^{u^{\star}}$ and $G^{u^{\bullet}}$. Without lose of generality we consider that $S=\{v: (v,u)\in D'\}$ and $S'=\{ u^{\star},u^{\bullet}\}$. Then, the set  $D'$ satisfies one of the following conditions:
\begin{enumerate}
\item [a)]
By Lemma \ref{Lex1}, $D=\{(v,u^{\star}):v\in S\}$, so $S$ is a total $[1,2]$-set for $G$ and if there exists a vertex $v\in D$ such that $\vert N(v) \cap S\vert =2$, then $H$ has an isolated vertex.
  \item [b)]
  $D'=\{(v,u^{\star}):v\in S\;\text{and}\;u\in S'\}$, by Corollary \ref{Lex2}, $S$ is an efficient dominating set of $G$ and $S'$ is  a total $[1,2]$-set for $H$.
  \item [c)]
  There is a vertex $w\in S$ such that $(w,u^{\star})\in D'$ but $(w,u^{\bullet})\notin D'$. By Lemma \ref{Lex1and2}, we have $S\in \mathcal{SD}^1_{[1,2]}(G)$ and $S'$ is  a total $[1,2]$-set for $H$.
\end{enumerate}

Now, we show the other side as follows:

\begin{enumerate}
\item
If $G=K_1$ and $H$ has a total $[1,2]$-set $S'$, then it is easy to see that $G \circ H=H$ and $S'$ is a  total $[1,2]$-set of $G \circ H$.
  \item
  Assume  that $S$ is a total $[1,2]$-set of $G$ and $u^{\star}\in V(H)$. We define $D$ as $S\times \{u^{\star}\}$. Since every vertex of $G^ {u^{\star}}$  is dominated by at least one of vertices of $D$, then every vertex of other $G_{-}$Layers is dominated by $D$. So, for any vertex $(v',u')\in G \circ H$, we have $\vert N((v',u')) \cap D\vert  \geq 1$. Now, it is sufficient to show that $\vert N((v',u')) \cap D\vert  \leq 2$. To this end, we  consider two cases:
  \begin{enumerate}
\item [a)]
For every vertex $v\in S$, $\vert N(v) \cap S\vert =1$: So, it is clear that for any vertex  $(v',u^{\star})$ of $G^ {u^{\star}}$, $\vert N((v',u^{\star})) \cap D\vert  \leq 2$. If $u'\neq u^{\star}$, we need to show that $\vert N((v',u')) \cap D\vert  \leq 2$. Then following cases can happen:
   \begin{itemize}
\item [a1)]
  $(v',u^{\star})\in D$ and $\{u',u^{\star}\}\in E(H)$; for every $v''\in S$  adjacent to  $v'$, $(v',u')$  is dominated by  $(v',u^{\star})$ and $(v'',u^{\star})$. Since $(v',  u^{\star})\in D$ and $v' \in S$, so $\vert N(v') \cap S\vert =2$ and $\vert N((v',u')) \cap D\vert =\vert N(v') \cap S\vert +1=2$.
\\
\item [a2)]
$(v',u^{\star})\in D$ and $\{u',u^{\star}\}\notin E(H)$; if $v''\in S$ and $\{v',v''\}\in E(G)$  then $(v',u')$ is dominated by  $(v'',u^{\star})$. So $\vert N((v',u')) \cap D\vert =\vert N(v') \cap S\vert =1$.
\\
 \item [a3)]
$(v',u^{\star})\notin D$; for every $v''\in S$ and $\{v',v''\}\in E(G)$, $(v',u')$ is dominated by  $(v'',u^{\star})$. Since $(v',u^{\star})\notin D$, $v'\notin S$. We have$\vert N((v',u')) \cap D\vert =\vert N(v') \cap S\vert \leq 2$.
\\
\end{itemize}

  \item [b)]
  There is a vertex $v\in S$ such that $\vert N(v) \cap S\vert =2$ and $u^{\star}$ is an isolated vertex in $H$. For every vertex $v''\in S$ and $\{v',v''\}\in E(G)$,  $(v',u')$  is dominated by $(v'',u^{\star})$. So it is the case that $\vert N((v',u')) \cap D\vert =\vert N(v') \cap S\vert \leq 2$.
  \end{enumerate}

  \item
 Let $S$ be an efficient dominating set of $G$, $S'=\{ u^{\star},u^{\bullet}\}$  is a total $[1,2]$-set for $H$ and $D=\{(v,u): v\in S\;and\;u\in S'\}$. It is easy to see that $D$ is a total dominating set of $G \circ H$.
 \\
  If $v'\in S$, then every  $(v',u') \in V(H^{v'})$ are dominated by  either $(v', u^{\star})$ or $(v', u^{\bullet})$. Since $S$ is an efficient dominating set of $G$, then $N_G(v')\cap S= \emptyset$ and $(v',u')$ is not dominated by any other vertices.
   If $v'\notin S$, then there is exactly one vertex $v''\in S$  such that $\{v',v''\}\in E(G)$ and every  $(v',u') \in V(H^{v'})$ are dominated by  either $(v'', u^{\star})$ and $(v'', u^{\bullet})$. So, $D$ is a total $[1,2]$-set for $G \circ H$.
  \item
Suppose that $S\in \mathcal{SD}^1_{[1,2]}$, $S'=\{ u^{\star},u^{\bullet}\}$ is a total $[1,2]$-set for $H$ and
$$D=\{(v,u^{\star}),(v,u^{\bullet}):v\in S\;\text{and} \;\vert  N(v)\cap S\vert  =0 \} \cup \{(v,u^{\star}):v\in S\;\text{and} \;\vert  N(v)\cap S\vert  =1 \} .$$
\\
 By definition of $D$, It is easy to see that for any vertex $(v,u) \in D$, there is a vertex $(v',u')\in D$ such that $\{(v,u),(v',u')\} \in E(G\circ H)$. So, $D$ is a total set of $G \circ H$. Now, we must show that $D$ dominates all vertices of  $G \circ H$ at least one and at most two times. It is clear $S=\{ v: (v,u^{\star})\in D\}\in \mathcal{SD}^1_{[1,2]}$. We consider three kind of vertices and we will show vertices of each $H_{-}$Layer are dominated by at least one and two vertices of $D$.
\begin{enumerate}
\item [a)] $v \in S$ and $\vert  N(v)\cap S\vert  =0$:
Since $S'=\{ u^{\star},u^{\bullet}\}$ is a total $[1,2]$-set for $G \circ H$, $(v,u^{\star})\in D$ and $(v,u^{\bullet})\in D$. Then, all of the vertices of  $H^v$ are dominated by   $(v,u^{\star})$ and $(v,u^{\bullet})$. Since $\vert  N(v)\cap S\vert  =0$. So, any other vertex cannot dominate vertices of $H^v$. Therefore $1\leq\vert  N(v,u)\cap D\vert  \leq 2$.
\\
\item [b)] $v \in S$ and $\vert  N(v)\cap S\vert  =1$:
 So, there is a vertex $v' \in S$ such that $\{v,v'\} \in E(G)$, $(v', u^{\star})$  dominates all of the vertices of $H^v$ and these vertices can also be dominated by  $(v, u^{\star})$. Since $S$ is a $1$-dependent $[1,2]$-set for $G$, then there is not any other vertex in neighborhood of $v$ in $S$, so $1\leq\vert  N(v,u)\cap D\vert  \leq 2$.
\\
\item [c)] $v \notin S$: Since $S$ is a $1$-dependent $[1,2]$-set for $G$, it is easy to see that  there is a vertex $v' \in S$ such that $\{v,v' \}\in E(G)$. So, all of the vertices of $H^v$ are dominated by $(v',u^{\star})$. If $\vert  N(v')\cap S\vert  =0$, then $(v',u^{\bullet})$ dominates vertices of $H^v$ and any other vertices can not dominate them. If there exist a $v'' \in S$ such that $\{v,v''\} \in E(G)$ and it is contradict to $dist_G(v',v'') \geq 3$. If $\vert  N(v')\cap S\vert  =0$, there maybe exists a vertex $(v'',u^{\star}) \in D$ such that  $\vert  N(v')\cap S\vert  \neq 0$ and there is no vertex in $H^{v''}$ and other $H_{-}$Layers dominate vertices of $H^v$.
\end{enumerate}

\end{enumerate}

\end{proof}
\subsection{Total $[1,k]$-set of Lexicographic Product of Graphs}
In this section, we express necessary and sufficient conditions for the given graphs $G$ and $H$ such that $G \circ H$ has a total  $[1,k]$-set. The Lemma \ref{Lex3}, \ref{Lex1}, \ref{Lex1and2} and Corollary \ref{Lex2} are generalized to total $[1,k]$-set. Since proofs in this section can be similarly obtained from the case on total $[1,2]$-sets, we omit them.

\begin{theorem}\label{Lexk}
Let $D$ be a total $[1,k]$-set for $G \circ H$.
\begin{enumerate}
  \item [a)]
  If $D$ contains more than $k$ vertices of a $H_{-}$Layer, then $G=K_1$ and $H \in \mathcal{D}^t_{[1,k]}$.

 \item [b)]
 If $D$ contains at most one vertex of every $H_{-}$Layers, then $S=\{v\in V(G):(v,u)\in D\}$ is a $(k-1)$-dependent total $[1,k]$-set of $G$. Moreover if there is a vertex $v \in S$ such that $\vert N(v) \cap S\vert =k$, then $H$ contains an isolated vertex.

 \item [c)]
If $H$ does not contain any isolated vertex and  $S=\{v\in V(G):(v,u)\in D\}$ is not a total set of $G$, then $D$ contains at most $k$ vertices of each $H^v$ and  satisfies the following conditions:

\begin{enumerate}
  \item [c1)]
  The set $S'=\{u\in V(H):(v,u)\in D\}$ is a  total $[1,k]$-set  of $H$ with cardinality to at most $k$  and there is a vertex $x\in S$  such that $1<|D\cap V(H^x)|\leq |S'|$;

\item [c2)]
$S$ is a $(k-1)$-dependent $[1,k]$-set for $G$;
\item [c3)]
 If there exist a vertex  $v\in S$ such that $\vert  N(v)\cap S\vert  =0$, then  $1<|D\cap V(H^v)|\leq \lfloor k/ 2 \rfloor$ or  for any vertex $v' \in S-\{v\}$, we have $dist_G(v,v')\geq 3$.
 \end{enumerate}
 \end{enumerate}

 \end{theorem}

\begin{theorem}
	Let $G$ and $H$ be two graphs. $G\circ H \in \mathcal{D}^t_{[1,k]}$ if and only if $G$ and $H$ satisfy one of the following conditions

\begin{enumerate}
\item
$G=K_1$ and $H \in \mathcal{D}^t_{[1,k]}$;
  \item
  $G$ has a total $[1,k]$-set $S$ and if $S$ has a vertex $v$ such that $\vert N(v) \cap S\vert =k$ then $H$ has an isolated vertex;
  \item
  $G$ is an efficient domination graph  and $\gamma_{t[1,k]}(H)\leq k$;
  \item
   $G$ has a $(k-1)$-dependent $[1,k]$-set $S$ and if $S \in \mathcal{SD}^{k-1}_{[1,k]}(G)$ then $\gamma_{t[1,k]}(H)\leq k$ and otherwise $\gamma_{t[1,k]}(H)\leq k/2.$

\end{enumerate}
\end{theorem}

\subsection{Some Result in independent $[1,k]$-set for Lexicographic Product of Graphs}
For two given integers $j$ and $k$, an independent subset $D \subseteq V$ is called an independent $[j, k]$-set  if for every vertex $v \in V\setminus D$, we have $j \leq \vert N(v) \cap D \vert \leq k$. The $[j, k]$-independence number of $G$ is the minimum number among cardinalities of independent $[j, k]$-sets in $G$ and is denoted by $\gamma_{i[j,k]}(G)$. The class of all graphs $G$ having at least one independent $[j,k]$-set is denoted by  $\mathcal{D}_{[j,k]}^i$.

\begin{theorem} \label{lexthem2}
	Let  $G$ and $H$ be two graphs. Then, $G \circ H \in \mathcal{D}_{[1,2]}^i$  if and only if one of the following conditions is satisfied:
\begin{enumerate}
\item
$G=K_1$ and $H \in \mathcal{D}^i_{[1,2]}$;
  \item
  $G$ is an efficient domination graph and $\gamma_{i[1,2]}(H)\leq2$;
\item
  $G \in \mathcal{D}_{[1,2]}^i$ and $\gamma_{i[1,2]}(H)=1$.
\end{enumerate}

\end{theorem}
\begin{proof}
The proof is similar to the proof of Theorem \ref{lexthem}.
\end{proof}
We also generalize Theorem \ref{lexthem2} to  independent $[1,k]$-set of  $G \circ H$.

\begin{theorem}
	Let $G$ and $H$ be two graphs. Then, $G\circ H \in \mathcal{D}^i_{[1,k]}$ if and only if one of the following conditions is satisfied:
\begin{enumerate}
\item
$G=K_1$ and $H \in \mathcal{D}^i_{[1,k]}$;
  \item
  $G$ is an efficient domination graph  and $\gamma_{i[1,k]}(H)\leq k$;
  \item
   $G$ has a independent $[1,k]$-set  and $\gamma_{i[1,k]}(H)\leq k/2$.
   \end{enumerate}
   \end{theorem}

\section{On $[1,2]$-domination number of Lexicographic Products of Graphs} \label{lexproductno}
In this section, we first describe the relation between  the domination and total domination number of  $G \circ H$ with respect to domination and total domination number of its components. Then, we use this relationship to compute $\gamma_{[1,2]}(G\circ H)$ and $\gamma_{t[1,2]}(G\circ H)$. At the end of this section, we generalize results to $\gamma_{[1,k]}(G\circ H)$ and $\gamma_{t[1,k]}(G\circ H)$.
\begin{theorem} \label{lexthem}
	For two arbitrary graphs  $G$ and $H$,
\begin{equation*}
\gamma(G \circ H)= \left\{
\begin{array}{ll}
 \gamma(G) & \text{if } \gamma(H)=1;\\
 \gamma_t(G) & \text{if } G \text{ has a total dominating set};\\
\vert V(G)\vert \times \vert V(H) \vert & \text{otherwise}.
\end{array} \right.
\end{equation*}
\end{theorem}
\begin{proof}

Let $D$ be a $\gamma$-set of $G\circ H$ and $S=\{v:(v,u)\in D\}$.
If $S$ is not a dominating set of $G$, then there is a vertex $v' \in V  \setminus  S$  which is not dominated by $S$.
It is easy to see that there is no vertex $(v',u')\in D$ such that $\{v,v'\}\in E(G)$.  Hence, vertices of $V(H^{v'})$ are dominated by itself. So,  $v' \in S$ and which is a contradiction.
Therefore, $S$ is a dominating set of $G$ and  $\gamma(G\circ H) \geq  \gamma(G)$.

Suppose that $\gamma(H)=1$ and $S$ is a $\gamma$-set for $G$.  Then, there exists a vertex $u \in V(H)$ such $u$ dominates all vertices of $H$ and $D=\{(v,u):v\in S\}$  is a $\gamma$-set for $G\circ H$.   So  $\gamma(G\circ H) =  \gamma(G)$.

If $\gamma(H)>1$, since $D$ is a dominating set of $G \circ H$, every vertex like $(v',u')\in V(H^{v'})$ is dominated by $D$. Assume that there is a vertex $v \in S$ such that $ \vert N_G(v) \cap S \vert=0$. Then, there is no vertex like $(w,u)\in D$ such that $\{v,w\}\in E(G)$. By  $\gamma (H)>1$, there are at least two vertices of $V(H^v)$ in $D$ such as $(v,u)$ and $(v,u')$.  By lemma \ref{Lexneighbour}, if $\{v,v'\}\in E(G)$ then $D'=(D\setminus V(H^v))\cup\{(v,u),(v',u)\}$ is a dominating set of $G\circ H$.  Clearly, $\vert D' \vert \leq \vert D \vert$  and $S'=\{v:(v,u)\in D'\}$ is a total dominating set of $G$. Therefor, we have $\vert  D \vert=\vert D' \vert=\gamma_t(G)$.
\end{proof}

\begin{theorem} \label{lexthemt}
	Let  $G$ and $H$ be two graphs. Then $G \circ H$ has a total dominating set if and only if $G$ has a total dominating set. In addition
$$	\gamma_t(G \circ H)=\gamma_t(G). $$
\end{theorem}

\begin{proof}
It is known that $\gamma(G)\leq \gamma_t(G )$. By Theorem \ref{lexthem}, if $ G \in D_{[1,2]}^t$ and $\gamma(H)>1$, then $\gamma_t(G \circ H)=\gamma_t(G )$.
Suppose that $\gamma(H)=1$ and $u^{\star}$ dominates all vertices of $H$.
It is easy to see that for every $\gamma$-set $S$ of $G$,   $D=\{(v,u^{\star}):v\in S \} $ is a total set of $G \circ H$. So
\begin{equation} \label{eq1}
\gamma_t(G\circ H) \leq \gamma_t(G).
\end{equation}
Let $D$ be  a $\gamma_t$-set of $G\circ H$  and $S'=\{v:(v,u)\in D\}$ is not a total set of $G$. Then, there exists a vertex $x\in S'$ and two adjacency vertices $y,y'\in V(H)$ such that $\vert N(x)\cap S'\vert =0$ and $(x,y), (x,y')\in D$.
Similar to proof of Theorem \ref{lexthem},  we can remove $(x,y')$ from $D$ and add $(x',y)$ to $D$ such that $\{x,x'\}\in E(G)$. For all $v\in V(G)$ where $\vert  D \cap V(H^v) \vert \geq 2$, we can do this process  and construct a new total dominating set $D'$ for $G\circ H$  such that its cardinality is not more than $D$  and $\vert V(H^v) \cap D' \vert =1$. Since the set $S''=\{v:(v,u)\in D\}$ is a total set of $G$ and  $\vert D\vert$ is a $\gamma_t$-set of $G \circ H$, then we have

 \begin{equation} \label{eq2}
 \gamma_t(G \circ H) \geq \vert S'' \vert \geq \gamma_t(G ).
 \end{equation}

  Therefore, by Equations \ref{eq1} and \ref{eq2}, we have $\gamma_t(G\circ H) = \gamma_t(G)$.

\end{proof}
	In Lemma \ref{Lex3}, we have shown  that for a nontrivial graph $G$, every total $[1,2]$-set of $G \circ H$
	 contains  at most two vertices of each $H_{-}$Layer. We generalize this result for  $[1,2]$-set of $G \circ H$.
	
	 \begin{lemma}
	 Let $G$ and $H$ be two nontrivial graphs. Then every $[1,2]$-set of $G \circ H$ contains  at most two vertices of each $H_{-}$Layer or $\gamma_{[1,2]}(G \circ H)=\vert V(G)\vert \times\vert V(H) \vert $ but not both.
	 	\end{lemma}

\begin{theorem} \label{lexthem2}
Let $G$ and $H$ be two graphs. Then, $\gamma_{[1,2]}(G\circ H)$ can be computed as follow:
\begin{description}
 \item[Case 1:]
  $H$ has  an isolated vertex:
 \begin{enumerate}
 \item [a)]
  If $G\in \mathcal{D}^t_{[1,2]}$,  then $\gamma_{[1,2]}(G \circ H)=\gamma_{t[1,2]}(G)$;
  \item [b)]
  If   $\gamma_{[1,2]}(H)=2$ and $S \in \mathcal{SD}^{2}_{[1,2]}(G)$, then $\gamma_{[1,2]}(G \circ H)=\min\{\vert S\vert +\alpha \}$ where $\alpha$ is the number of vertices in $S$ such that $\vert N(v)\cap S\vert =0$;
  \item [c)]
  Otherwise, $\gamma_{[1,2]}(G \circ H)=\vert V(G)\vert \times \vert V(H)\vert $.
  \end{enumerate}
  \item[Case 2:]
   $H$ does not have  an isolated vertex:
 \begin{enumerate}
 	\item [a)]
 	If $\gamma_{[1,2]}(H)=1$ and  $S$ is a $1$-dependent $[1,2]$-set of $G$ with minimum cardinality,  then $\gamma_{[1,2]}(G \circ H)=\vert S\vert $;
 	\item [b)]
 	If $S$ is a $1$-dependent total $[1,2]$-set of $G$  with minimum cardinality, then $\gamma_{[1,2]}(G \circ H)=\vert S\vert $;
  \item [c)]
  If $\gamma_{[1,2]}(H)=2$ and $S\in \mathcal{SD}^{1}_{[1,2]}(G)$, then $\gamma_{[1,2]}(G \circ H)=\min\{\vert S\vert +\alpha \}
  $ where $\alpha$ is the number of vertices in $S$ such that $\vert N(v)\cap S\vert =0$;
  \item [d)]
  Otherwise, $\gamma_{[1,2]}(G \circ H)=\vert V(G)\vert \times \vert V(H)\vert $.
\end{enumerate}
\end{description}
\end{theorem}

\begin{proof}
	We just show the first case, since the second one is easily obtained given the first.
	
Let $u^{\star}$ be an isolated vertex of $H$.
We claim that if none of the following conditions are met, then $\gamma_{[1,2]}(G \circ H)=\vert V(G)\vert \times \vert V(H)\vert $:
\begin{enumerate}
	\item[a)]
 $G \in D_{[1,2]}^t$;
 \item[b)]
  $\gamma_{[1,2]}(H)\leq 2$ and $\mathcal{SD}^2_{[1,2]}(G)\neq \emptyset$.
\end{enumerate}

	If $S$ is a total $[1,2]$-set $S$ of $G$, then $D=\{(v,u^{\star}):v\in S\}$ is a total $[1,2]$-set for $G\circ H$.
	If $G \notin D_{[1,2]}^t$ and $D$ is a $[1,2]$-set of $G\circ H$, then recall that $S'=\{v:(v,u)\in D \}$ is a $[1,2]$-set for $G$. Assume $S'$ is not $2$-dependent, i.e. there is a vertex $x\in S'$ such that $\vert N(x)\cap S' \vert >2$. It is easy to verify that $\vert V(H^x)\cap D\vert >2$. Hence, $V(H^x)\subseteq D$. If $\vert V(H) \vert >2$ then for every vertex $w \in N(v)$, we have $V(H^w) \subseteq D$. We continue this approach and since $V(G)$ is finite, it will terminate as soon as all vertices are visited. So for all $v\in V(G)$ we have $V(H^v)\subseteq D$. Therefore, we obtain  $\gamma_{[1,2]}(G \circ H)=\vert V(G)\vert \times \vert V(H)\vert $.
In addition, for every $v \in S$ such that $\vert N(v)\cap S\vert =0$,  $V(H^v)$ must be dominated by at most two vertices of $V(H^v)$. Therefore, $\gamma_{[1,2]}(H)=2$. Now, we have two cases to consider.

\begin{enumerate}
\item [b1)]
	For every $\gamma_{t[1,2]}$-set $S$ of $G$, $D=\{(v,u^{\star}):v\in S\}$ is a $\gamma_{t[1,2]}$-set for $G\circ H$. See Lemma \ref{lexthem};
	\item [b2)]
	Let $\{u^{\star},u^{\bullet}\}$ be a $[1,2]$-set for $H$,  $G \notin D_{[1,2]}^t$ and for every $2$-dependent $[1,2]$-set $S$ of $G$,  $S_{\alpha}$ is defined as $\{v:N(v)\cap S=\emptyset\}$. Vertices of $\bigcup_{x\in S\setminus S_{\alpha}}V(H^x)$ are dominated by $(x',u^{\star})$ where $\{x,x'\}\in E(G)$, and  vertices of $\bigcup_{x\in S_{\alpha}}V(H^x)$ are dominated by $(x,u^{\star})$ and $(x,u^{\bullet})$. So
$\{(v,u^{\star}):v \in S\}\cup\{(v,u^{\bullet}):v \in S\wedge\vert N(v)\cap S\vert  =0\}$  is a $\gamma_{[1,2]}$-set for $G \circ H$ if and only if $S\in \mathcal{SD}^2_{[1,2]}(G)$.
\end{enumerate}
\end{proof}

\begin{theorem} \label{lexthemt2}
Let $G$ and $H$ be two graphs. Then,  $\gamma_{t[1,2]}(G\circ H)$ can be computed as follow:
\begin{description}
	\item [Case 1:]
	 $H$ has  an isolated vertex:
	\begin{enumerate}
		\item [a)]
  If $G\in \mathcal{D}^t_{[1,2]}$,  then $\gamma_{t[1,2]}(G \circ H)=\gamma_{t[1,2]}(G)$;
  \item [b)]
Otherwise, $G\circ H\notin \mathcal{D}^t_{[1,2]}$.
  \end{enumerate}
   \item [Case 2:]
 	 $H$ does not have  an isolated vertex, then one of the following conditions holds.
 \begin{enumerate}
  \item [a)]
   If $S$ is a $1$-dependent total $[1,2]$-set of $G$ with minimum cardinality, then $\gamma_{t[1,2]}(G \circ H)=\vert S\vert $;
  \item [b)]
   If $\gamma_{[1,2]}(H)=2$ and $S\in \mathcal{SD}^{1}_{[1,2]}(G)$ then $\gamma_{t[1,2]}(G \circ H)=min\{\vert S\vert +\alpha \}
  $ where  $\alpha$ is the number of vertices in $S$ such that $\vert N(v)\cap S\vert =0$;
  \item [c)]
  $G\circ H\notin \mathcal{D}^t_{[1,2]}$ and $\gamma_{t[1,2]}(G\circ H)$ is undefined since there is no total $[1,2]$-set for $G$.
\end{enumerate}
\end{description}
\end{theorem}
 \begin{proof}
 	The proof is similar to the proof of Theorem \ref{lexthem2}.
 \end{proof}

\begin{example}
Let $G_1$ and $G_2$ be two graphs shown in Figure \ref{fig:lex1}. $G_1\notin \mathcal{D}^t_{[1,2]}$, so for any graph $H$ such that $\gamma_{[1,2]}(H)>2$, $\gamma_{[1,2]}(G \circ H)=\vert V(G)\vert \times\vert V(H)\vert $. Since $G_2$  has a total $[1,2]$-set but it does not have any $1$-dependent total $[1,2]$-set, so for any graph $H$ without isolated vertices such that $\gamma_{[1,2]}(H)>2$, we have  $\gamma_{[1,2]}(G \circ H)=\vert V(G)\vert \times\vert V(H)\vert $.
\begin{figure}[h!]
   \centering
     \includegraphics[width=0.4\textwidth]{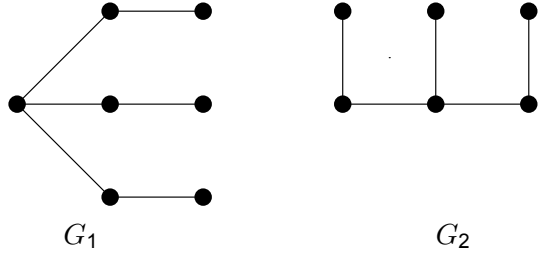}
      \caption{$G_1 \notin \mathcal{D}^t_{[1,2]}$ and $G_2 \in \mathcal{D}^t_{[1,2]}$}
      \label{fig:lex1}
\end{figure}
\end{example}
In the sequel, we compute $\gamma_{i[1,2]}(G \circ H)$  and $\gamma_{i[1,k]}(G \circ H)$ which can be proved  as Theorem \ref{lexthem2}.

\begin{theorem} \label{lexthemind}
	Let $G$ and $H$ be two graphs.
 \begin{enumerate}
  \item [a)]
   If $G$  has  an efficient dominating and $H$  has an independent $[1,2]$-set of size of at most $2$, then $\gamma_{i[1,2]}(G \circ H)= \gamma_{\text{efficient}}(G) \times \gamma_{i[1,2]}(H) $;
  \item [b)]
  If $G$  has an independent $[1,2]$-set  and $H$  has an independent $[1,2]$-set  like $S'$ such that $\vert S' \vert=1$, then $\gamma_{i[1,2]}(G \circ H)=\gamma_{i[1,2]}(G)$;
  \item [c)]
  Otherwise, $G \circ H$ has not independent $[1,2]$-sets.
\end{enumerate}
\end{theorem}

\begin{theorem} \label{lexthemindk}
	Let $G$ and $H$ be two graphs.
 \begin{enumerate}
  \item [a)]
   If $G$  has  an efficient dominating set and $H$  has a independent $[1,k]$-set of size of at most $k$, then $\gamma_{i[1,k]}(G \circ H)= \gamma_{\text{efficient}}(G) \times\gamma_{i[1,k]}(H) $;
  \item [b)]
  If $G$  has an independent $[1,k]$-set and $H$  has an independent $[1,k]$-set of size at most $\lfloor k/2 \rfloor$, then $\gamma_{i[1,k]}(G \circ H)= \gamma_{i[1,k]}(G) \times \gamma_{i[1,k]}(H) $;
  \item [c)]
  Otherwise $G \circ H$, has not independent $[1,k]$-sets.
\end{enumerate}
\end{theorem}

For nontrivial path and cycles, by results for lexicographic products of graph, the following results can be obtained.
 \begin{corollary}\label{a}
 Let $P_n$ and $P_m$ be two  nontrivial paths. Then,
\begin{equation*}
\gamma_{[1, 2]}(P_n \circ P_m)  =\left\{
\begin{array}{ll}
\lceil \frac{n}{3}\rceil\; & \text{for } m=2,3;\\
2\lceil \frac{n}{4}\rceil\; & \text{for } m>3.\\
\end{array} \right.
\end{equation*}
$$\gamma_{t[1, 2]}(P_n \circ P_m)  =2\left\lceil \frac{n}{4}\right\rceil.$$

\begin{equation*}
\gamma_{i[1, 2]}(P_n \circ P_m)  =\left\{
\begin{array}{ll}
\lceil \frac{n}{3}\rceil\; & \text{for } m=2,3;\\
2\lceil \frac{n}{3}\rceil\; & \text{for } m=4,5,6;\\
\text{not exit} \; & \text{for } m>6.\\
\end{array} \right.
\end{equation*}
\end{corollary}

\begin{corollary}\label{b}
 Let $C_n$ and $C_m$ be two arbitrary nontrivial path, then
\begin{equation*}
\gamma_{[1, 2]}(C_n \circ C_m)  =\left\{
\begin{array}{ll}
\lceil \frac{n}{3}\rceil\; & \text{for } m=2,3;\\
5m& \text{for } n=5;\\
2\lceil \frac{n}{4}\rceil\; & \text{for } m>3 \text{ and } n\neq5 .\\
\end{array} \right.
\end{equation*}

\begin{equation*}
\gamma_{t[1, 2]}(C_n \circ C_m)  =\left\{
\begin{array}{ll}
\text{not exit} \; & \text{for } n=5;\\
2\lceil \frac{n}{4}\rceil\; & \text{for } n\neq 5.\\
\end{array} \right.
\end{equation*}

\begin{equation*}
\gamma_{i[1, 2]}(C_n \circ C_m)  =\left\{
\begin{array}{ll}
\lceil \frac{n}{3}\rceil\; & \text{for } m=2,3;\\
2\lceil \frac{n}{3}\rceil\; & \text{for } m=4,5,6 \text{ and } n\equiv 0\;\mod 3;\\
\text{not exit} \; & \text{otherwise } .\\
\end{array} \right.
\end{equation*}
\end{corollary}
The values obtained in Corollary \ref{a} are exactly the same as  values for $P_n \circ C_m$ and the ones in Corollary  \ref{b} are exactly as the  values for $C_n \circ P_m$.

\section{Complexity} \label{complextysec}

In this section, we will show that the decision problem for total $[1,2]$-set is $NP$-complete. We will do this  by reduction the $NP$-complete problem, Exact $3$-Cover, to Total $[1, 2]$-Set.\\

\textbf{Exact $3$-cover problem:}
\\
Input of this problem is a finite set $X=\{ x_1, x_2, .... , x_{3q} \}$ with $|X| = 3q$ and a collection $C$ of 3-element subsets of $X$ such as  $C_i =\{x_{i_1},x_{i_2},x_{i_3}\}$. our goal is to understand is there a $C'\subseteq C$ such that every element of $X$ appears in exactly one element of $C'$?\\

\textbf{Total $[1, 2]$-set problem:}
\\
Input of this problem is a graph $G = (V, E)$ and a positive integer $k \leq |V|$.
We want to investigate is there any total $[1,2]$-set of cardinality at most $k$ for $G$.

\begin{theorem}\label{Complexity2}
	Total [1, 2]-SET is $NP$-complete for bipartite graphs.
\end{theorem}	
\begin{proof}
	Let $D\subseteq V$ is given, we verify $D$ is a total $[1,2]$-set. For any vertex $v\in D$, we check neighborhood of  each vertex and compute span number of any vertex $v \in V$. If there is a vertex $v$ with span number more than 2, this set isn't a   total $[1,2]$-set for $G$.
	It is obvious this algorithm is done in polynomial time and  total $[1,2]$-set is a $NP$ problem.
	Now for a set $X$, and a collection $C$  of 3-element subsets of $X$, we build a graph and transform EXACT 3-COVER  into a total $[1,2]$-set problem.
	\\	
	Let $X = \{x_{1}, x_{2}, . . . , x_{3q}\}$ and $C = \{C_{1}, C_{2}, . . . , C_{t}\}$.
	For each $C_{i}\in C$, we build a cycle $C_{4}$ with a vertex $u_{i}$. we add new vertices  $\{v_{1_1}, v_{1_2}, v_{1_3}, v_{2_1},v_{2_2},v_{2_3}, ..., v_{t_1},v_{t_2},v_{t_3}\}$ and connect all vertices $v_{i1},v_{i2},v_{i3}$ to $u_i$. Then add some other vertices $\{x_{1}, x_{2}, . . . , x_{3q}\}$ and edges $x_{i}v_{j_1}$, $x_{i}v_{j_2}$ and $x_{i}v_{j_3}$,  if $x_{i}\in C_{j}$. $G$ is a bipartite graph.
	\\
	Let $k = 2t + q$.
	Suppose that $C'$ is a solution for set $X$ and collection  $C$ of EXACT $3$-COVER.
	We build a set $D$ of vertices of $G$ contain every $u_{i}$, $1 \leq i \leq t$,
	and another vertex of $C_{4}$ adjacent to $u_i$ and one of the $v_{j_1}$, $v_{j_2}$ or $v_{j_3}$ for each $C_{j} \in C'$.
	If $C'$ exists, then   it's cardinality is precisely q, and so $|D| = 2t + q = k$.
	We can check easily  that $D$ is a [1, 2]-total set of $G$.
	\\	
	Conversely, suppose that $G$ has a total $[1, 2]$-set $D$ with $|D| \leq 2t + q = k$. Then $D$ must contain two vertices of every $C_{4}$, in the best case we select $u_{i}$ and one of the vertices  in that adjacency in $C_{4}$. We select $2t$ vertices that dominate all vertices of cycles and all vertices of form $v_{i_1}$, $v_{i_2}$ or $v_{i_3}$ for $1 \leq i \leq t$. Since each $v_{i_j}$ dominates only three vertices of $\{x_{1}, x_{2}, . . . ,x_{3q}\}$
	We have to select exactly $q$ vertices of them, i.e. we select $q$ 3-element subsets of form $\{v_{i_1}, v_{i_2}, v_{i_3} \}$ and one element of each of them.
	Each of this $v_{i_j}$ correspond to a $C_i$ and  union of them is a exact cover for $C$.
\end{proof}

\begin{example}
	Let $C=\{C_1,C_2,C_3,C_4\}=\{\{x_1,x_2,x_4\},\{x_3,x_5,x_7\},\{x_4,x_5,x_6,x_7\},\{x_6,x_8,x_9\}\}$, corresponding graph was shown in Figure \ref{fig:npgraph2}.
	\begin{figure}[h!]
		\centering
		\includegraphics[width=0.7\textwidth]{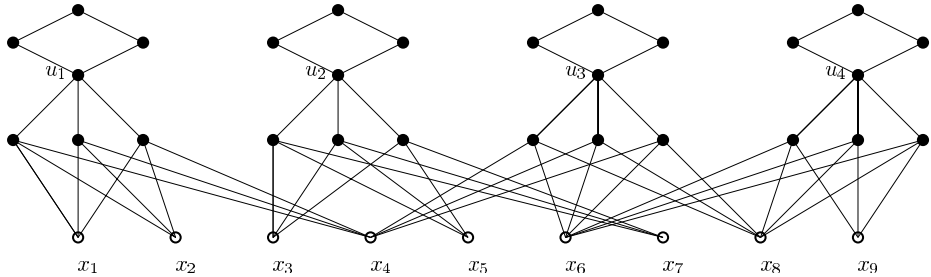}
		\caption{$NP$-completeness for bipartite graph}
		\label{fig:npgraph2}
	\end{figure}
\end{example}
 \section*{Acknowledgment}
        The authors are grateful to  Dr. A.~Shakiba and A. K. Goharshady for their constructive comments and suggestions on improving  our paper.


\begin{thebibliography}{10}
	
	\bibitem{allan1978domination}
	Robert~B Allan and Renu Laskar.
	\newblock On domination and independent domination numbers of a graph.
	\newblock {\em Discrete Mathematics}, 23(2):73--76, 1978.
	
	\bibitem{bange1988efficient}
	DW~Bange, AE~Barkauskas, and PJ~Slater.
	\newblock Efficient dominating sets in graphs.
	\newblock {\em Applications of Discrete Mathematics}, pages 189--199, 1988.
	
	\bibitem{bresar2008rainbow}
	Bostjan Bresar, Michael~A Henning, and Douglas~F Rall.
	\newblock Rainbow domination in graphs.
	\newblock {\em Taiwanese Journal of Mathematics}, 12(1):pp--213, 2008.
	
	\bibitem{Chang201289}
	Gerard~J Chang, BS~Panda, and D~Pradhan.
	\newblock Complexity of distance paired-domination problem in graphs.
	\newblock {\em Theoretical Computer Science}, 459:89--99, 2012.
	
	\bibitem{chellali2012k}
	Mustapha Chellali, Odile Favaron, Adriana Hansberg, and Lutz Volkmann.
	\newblock $k$-domination and $k$-independence in graphs: A survey.
	\newblock {\em Graphs and Combinatorics}, 28(1):1--55, 2012.
	
	\bibitem{chellali2014independent}
	Mustapha Chellali, Odile Favaron, Teresa~W Haynes, Stephen~T Hedetniemi, and
	Alice McRae.
	\newblock Independent $[1, k]$-sets in graphs.
	\newblock {\em Australasian Journal of Combinatorics}, 59(1):144--156, 2014.
	
	\bibitem{chellali20131}
	Mustapha Chellali, Teresa~W Haynes, Stephen~T Hedetniemi, and Alice McRae.
	\newblock $[1, 2]$-sets in graphs.
	\newblock {\em Discrete Applied Mathematics}, 161(18):2885--2893, 2013.
	
	\bibitem{dorbec2006some}
	Paul Dorbec, Sylvain Gravier, Sandi Klavzar, Simon Spacapan, et~al.
	\newblock Some results on total domination in direct products of graphs.
	\newblock {\em Discussiones Mathematicae Graph Theory}, 26(1):103--112, 2006.
	
	\bibitem{gavlas2002efficient}
	Heather Gavlas and Kelly Schultz.
	\newblock Efficient open domination.
	\newblock {\em Electronic Notes in Discrete Mathematics}, 11:681--691, 2002.
	
	\bibitem{goharshady20161}
	AK~Goharshady, MR~Hooshmandasl, and M~Alambardar Meybodi.
	\newblock [1, 2]-sets and [1, 2]-total sets in trees with algorithms.
	\newblock {\em Discrete Applied Mathematics}, 198:136--146, 2016.
	
	\bibitem{gravier1997domination}
	Sylvain Gravier and Michel Mollard.
	\newblock On domination numbers of cartesian product of paths.
	\newblock {\em Discrete Applied Mathematics}, 80(2):247--250, 1997.
	
	\bibitem{haas2014k}
	Ruth Haas and Karen Seyffarth.
	\newblock The $k$-dominating graph.
	\newblock {\em Graphs and Combinatorics}, 30(3):609--617, 2014.
	
	\bibitem{hammack2011handbook}
	Richard Hammack, Wilfried Imrich, and Sandi Klavzar.
	\newblock {\em Handbook of product graphs}.
	\newblock CRC press, 2011.
	
	\bibitem{hartnell2004dominating}
	Bert~L Hartnell and Douglas~F Rall.
	\newblock On dominating the cartesian product of a graph and $k_2$.
	\newblock {\em Discussiones Mathematicae Graph Theory}, 24(3):389--402, 2004.
	
	\bibitem{hartnell1991vizing}
	BL~Hartnell and DF~Rall.
	\newblock On Vizing's conjecture.
	\newblock {\em Congressus Numerantium}, pages 87--87, 1991.
	
	\bibitem{haynes1998fundamentals}
	Teresa~W Haynes, Stephen Hedetniemi, and Peter Slater.
	\newblock {\em Fundamentals of domination in graphs}.
	\newblock CRC Press, 1998.
	
	\bibitem{haynes1998domination}
	Teresa~W Haynes, Stephen~T Hedetniemi, and Peter~J Slater.
	\newblock {\em Domination in graphs: advanced topics}.
	\newblock Marcel Dekker, 1998.
	
	\bibitem{haynes1998paired}
	Teresa~W Haynes and Peter~J Slater.
	\newblock Paired-domination in graphs.
	\newblock {\em Networks}, 32(3):199--206, 1998.
	
	\bibitem{hedetniemi1991topics}
	Stephen~T Hedetniemi and RC~Laskar.
	\newblock {\em Topics on domination}.
	\newblock Elsevier, 1991.
	
	\bibitem{henning2013total}
	Michael~A Henning and Anders Yeo.
	\newblock Total domination and graph products.
	\newblock In {\em Total Domination in Graphs}, pages 103--108. Springer, 2013.
	
	\bibitem{hoory2006expander}
	Shlomo Hoory, Nathan Linial, and Avi Wigderson.
	\newblock Expander graphs and their applications.
	\newblock {\em Bulletin of the American Mathematical Society}, 43(4):439--561,
	2006.
	
	\bibitem{imrichklav}
	Wilfried Imrich and Sandi Klav{\v{z}}ar.
	\newblock {\em Product graphs, structure and recognition}, volume~56.
	\newblock Wiley-Interscience, 2000.
	
	\bibitem{kuziak2014efficient}
	Dorota Kuziak, Iztok Peterin, and Ismael~Gonzalez Yero.
	\newblock Efficient open domination in graph products.
	\newblock {\em Discrete Mathematics $\&$ Theoretical Computer Science},
	16(1):105--120, 2014.
	
	\bibitem{li2009total}
	Ning Li and Xinmin Hou.
	\newblock On the total $k$-domination number of cartesian products of graphs.
	\newblock {\em Journal of combinatorial optimization}, 18(2):173--178, 2009.
	
	\bibitem{nowakowski1996associative}
	Richard~J Nowakowski and Douglas~F Rall.
	\newblock Associative graph products and their independence, domination and
	coloring numbers.
	\newblock {\em Discussiones Mathematicae Graph Theory}, 16(1):53--79, 1996.
	
	\bibitem{rall2005total}
	Douglas~F Rall.
	\newblock Total domination in categorical products of graphs.
	\newblock {\em Discussiones Mathematicae Graph Theory}, 25(1-2):35--44, 2005.
	
	\bibitem{west2001introduction}
	Douglas~Brent West et~al.
	\newblock {\em Introduction to graph theory}, volume~2.
	\newblock Prentice hall Upper Saddle River, 2001.
	
	\bibitem{yang20141}
	Xiaojing Yang and Baoyindureng Wu.
	\newblock $[1, 2]$-domination in graphs.
	\newblock {\em Discrete Applied Mathematics}, 2014.
	
	\bibitem{yannakakis1980edge}
	Mihalis Yannakakis and Fanica Gavril.
	\newblock Edge dominating sets in graphs.
	\newblock {\em SIAM Journal on Applied Mathematics}, 38(3):364--372, 1980.
	
	\bibitem{Zhao201428}
	Yancai Zhao, Zuhua Liao, and Lianying Miao.
	\newblock On the algorithmic complexity of edge total domination.
	\newblock {\em Theoretical Computer Science}, 557:28--33, 2014.
	
\end{thebibliography}

\end{document}